\documentclass[12pt,english,journal]{elsarticle}
\usepackage[T1]{fontenc}
\usepackage[latin9]{inputenc}
\usepackage{geometry}
\usepackage{float}
\usepackage{amsthm}
\usepackage{amsmath}
\usepackage{graphicx}
\usepackage{setspace}
\makeatletter

\floatstyle{ruled}
\newfloat{algorithm}{tbp}{loa}
\providecommand{\algorithmname}{Algorithm}
\floatname{algorithm}{\protect\algorithmname}

\numberwithin{equation}{section}
\numberwithin{figure}{section}
\theoremstyle{plain}
\newtheorem{thm}{\protect\theoremname}
\theoremstyle{plain}
\newtheorem{lem}[thm]{\protect\lemmaname}
\ifx\proof\undefined
\newenvironment{proof}[1][\protect\proofname]{\par
\normalfont\topsep6\p@\@plus6\p@\relax
\trivlist
\itemindent\parindent
\item[\hskip\labelsep
\scshape
#1]\ignorespaces
}{%
\endtrivlist\@endpefalse
}
\providecommand{\proofname}{Proof}
\fi

\usepackage{times}

\@ifundefined{showcaptionsetup}{}{%
 \PassOptionsToPackage{caption=false}{subfig}}
\usepackage{subfig}
\makeatother

\usepackage{babel}
\providecommand{\lemmaname}{Lemma}
\providecommand{\theoremname}{Theorem}

\begin{document}

\title{Local Thresholding on Distributed Hash Tables}
\begin{abstract}
We present a binary routing tree protocol for distributed hash table
overlays. Using this protocol each peer can independently route messages
to its parent and two descendants on the fly without any maintenance,
global context, and synchronization. The protocol is then extended
to support tree change notification with similar efficiency. The resulting
tree is almost perfectly dense and balanced, and has $O\left(1\right)$
stretch if the distributed hash table is symmetric Chord. We use the
tree routing protocol to overcome the main impediment for implementation
of local thresholding algorithms in peer-to-peer systems -- their
requirement for cycle free routing. Direct comparison of a gossip-based
algorithm and a corresponding local thresholding algorithm on a majority
voting problem reveals that the latter obtains superior accuracy using
a fraction of the communication overhead.\end{abstract}
\maketitle

\section{Introduction}

In a world of millions of wired devices, in-network computation algorithms
provide an intriguing alternative to centralization. Where distributed
data is abundant and bandwidth is limited or costly, some applications
can only be implemented distributively. Where adverse manipulation
and control are a concern, distributed architecture is often preferred
over a centralized agent. Finally, scaling an algorithm to the millions
of peers often teaches important lessons on asynchrony, speculative
execution, and the containment of partial failure, which prove important
to more mundane environments such as grid systems.

Algorithms for distributed computation in peer-to-peer systems fall
into several categories. Of these, gossip algorithms are possibly
the most popular and certainly the most extensively studied \citep{KempeGossip,gossipBoyd,Gossip,gossipSpeedup,GossipVoting,gossipKalman,gossipKalman1,gossipSL,dynamicGossip}.
Local thresholding algorithms \citep{MajorityRulej,localRanking,l2,p2pDT,p2pEigen}
are comparable to gossip based algorithms because both address similar
problems and similarly provide a proof for convergence. Local thresholding
algorithms are considered by far more communication efficient than
gossip based algorithms. However, they pose far stricter requirements
to the underlying routing protocol. A gossip based algorithm basically
requires an efficient way in which information can be propagated to
random destinations. In contrast, all known local thresholding algorithms
require cycle free routing. Often, work on local thresholding algorithms
advocates that a routing tree be induced in preprocessing. However,
the non-trivial complexity of inducing and maintaining the tree in
a dynamic network has so far rendered local algorithms impractical.

This work considers the problem of computation in distributed hash-table
(DHT) overlays -- the de-facto standard architecture in peer-to-peer
networks. Gossip algorithms can easily be implemented on a DHT: If
each peer sends messages to a random peer from its finger table then
in $O\left(\log N\right)$ messages this  information will arrive
to a random peer. Local thresholding can be implemented in a DHT using
one of the existing tree routing protocols. However, existing tree
routing protocols \citep{DHTTree1,DHTTree2,dhttreeJ} are ill-fit
for a local thresholding algorithm. Because these protocols were developed
mainly to reduce message redundancy in broadcast or convergecast they
operate in a top-down or bottom-up manner. Thus, a peer cannot send
messages to its tree neighbors without the involvement of either the
root (in a top-down protocol) or its entire subtree (in a bottom-up). 

This paper makes two main contributions to the state-of-the-art: First,
it presents new binary tree routing and change notification protocols
for DHT overlays. This tree routing protocol is local and can be used
for multi-way communication over the tree, including broadcast and
convergecast. The effect of peer joining or leaving is also local
and can be detected and notified using no more than six messages that
are routed on the tree. The Enabled by the binary tree routing protocol,
the second contribution of this paper is a direct comparison of local
thresholding and gossip based algorithms. Our experiments show that
regardless of system size or properties of the data, local thresholding
vastly outperforms gossip. The results are so one sided that they
call into question the continued relevance of gossip algorithms to
computation in DHT overlays.

The rest of this paper is organized as follows: The next section describes
the binary tree routing protocol and the change notification protocol.
Section \ref{sec:Majority-voting} details the implementation of the
two majority voting algorithms. Experiments are described in Section
\ref{sec:Experimentation} and related work in Section \ref{sec:Related-Work}.
Finally, Section \ref{sec:Conclusions} draws conclusions and poses
some further research problems.

\section{\label{sec:Binary-Tree-Routing}Local Binary Tree Routing}

The basic idea of the binary tree routing protocol is to define a
mapping of peers to a subset of the nodes of a full binary tree. The
binary tree can be defined in terms of a one-to-one mapping of $d$-long
binary strings, namely addresses, to tree nodes. Then we define which
peer is mapped to which address. 

Consider a binary tree whose root is the all zero address. Any address
other than that of the root, we divide into three parts: An all zero
suffix, which might be empty, the rightmost set bit, and a prefix,
which might be empty as well. An address is therefore encoded as $p10^{k}$,
where the length of the prefix $p$ is $d-k-1$. We define that the
clockwise descendant of the address $p10^{k}$ is the address $p110^{k-1}$
and the counterclockwise descendant of the address $p10^{k}$ is $p010^{k-1}$.
Addresses ending with a set bit, i.e., $p10^{0}$, have no descendants.
For completeness, we denote $10^{d-1}$ the clockwise descendant of
the root. As can be seen in Figure \ref{fig:Tree-mapping}, this mapping
is similar, but not identical to, the textbook implementation of a
complete binary tree in an array.

\begin{figure*}
\caption{Binary Tree routing}

\begin{minipage}[t]{0.5\textwidth}%
\subfloat[\label{fig:Tree-mapping}Tree mapping to the address space.]{

\includegraphics[width=1\textwidth]{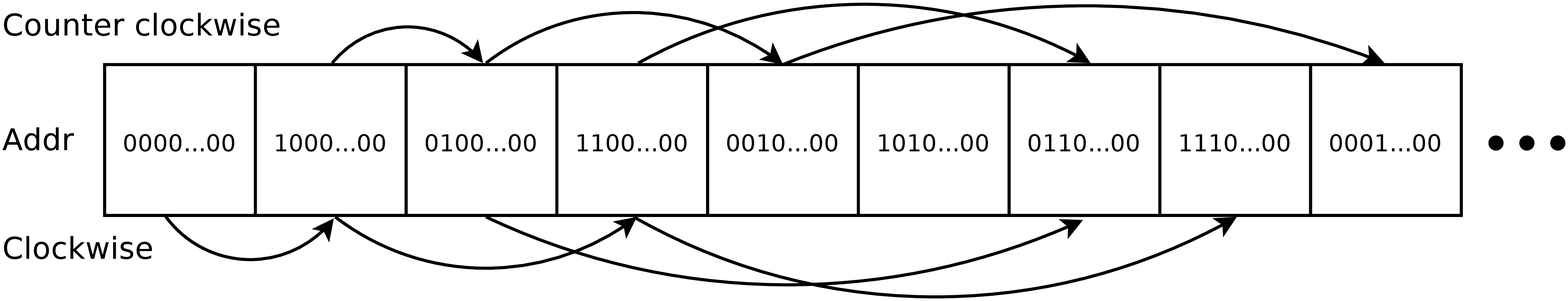}}%
\end{minipage}%
\begin{minipage}[t]{0.5\textwidth}%
\subfloat[\label{fig:Tree-on-DHT-1}Binary Tree Routing on DHT]{

\includegraphics[width=0.7\textwidth]{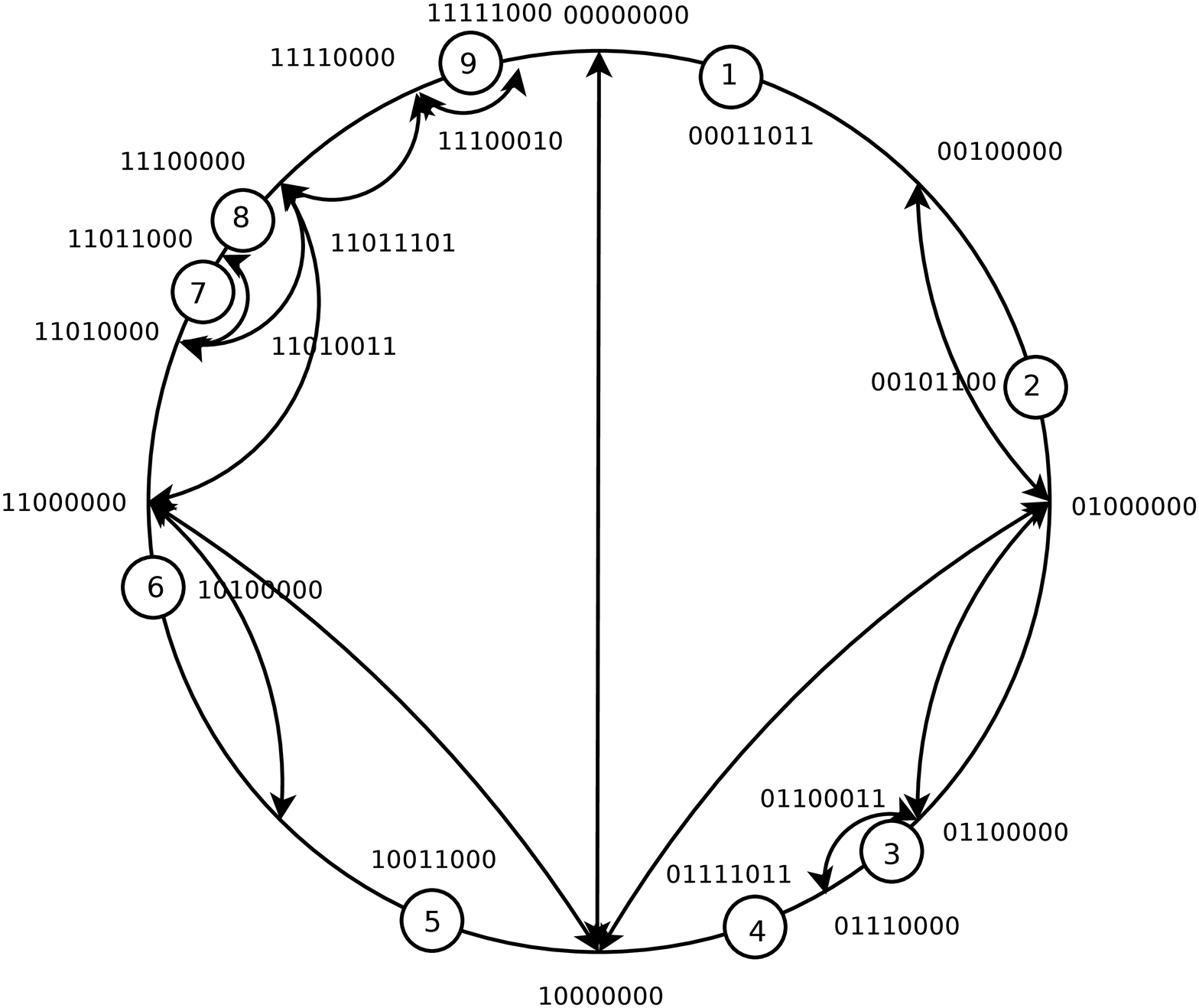}}%
\end{minipage}
\end{figure*}

A peer's position in the tree depends on its assigned address space
segment. The peer whose address space segment contains the all zero
address takes the root position. Any other peer takes a position calculated
as follows: Let the address space segment of $p_{i}$ be $\left(a_{i-1},a_{i}\right]$.
Let $p$ be the (possibly empty), common prefix of $a_{i-1}$ and
$a_{i}$, such that $a_{i-1}=p0X$ and $a_{i}=p1Y$. Then $p_{i}$
takes the position $p10^{k}$. Note that messages routed to the address
which is a peer's position will always be accepted by that peer. 

We conveniently denote the position with which $p_{i}$ is associated
as $pos_{i}$. We further denote the clockwise descendant of $pos_{i}$
as $CW\left[pos_{i}\right]$ and its counterclockwise descendant as
$CCW\left[pos_{i}\right]$. Respectively, if $pos_{j}=CW\left[pos_{i}\right]$
or $pos_{j}=CCW\left[pos_{i}\right]$, then we denote $pos_{i}=UP\left[pos_{j}\right]$.
The functions $CW$, $CCW$, and $UP$ can be computed for any $pos_{i}$
using bit manipulations. 

The following two lemmas show that if there is more than one peer
in the clockwise (or, respectively, the counterclockwise) subtree
of a peer's position then one of those peers occupies a position which
is a fore-parent of the positions of all other peers.
\begin{lem}
\label{lem:continuous}The address space segment associated with the
peers whose positions are the subtree below any peer $p_{i}$ is continuous.\end{lem}
\begin{proof}
See \ref{sec:Proofs}.\end{proof}
\begin{lem}
\label{lem:singularDescendant}For any peer $p_{i}$ whose position
is $pos_{i}$, one of the peers which occupies positions in the subtree
of position $CW\left[pos_{i}\right]$ (respectively, $CCW\left[pos_{i}\right]$)
occupies a position which is a fore-parent of the positions of all
of the other peers in that subtree.\end{lem}
\begin{proof}
See \ref{sec:Proofs}.
\end{proof}

Lemma \ref{lem:singularDescendant} can be put into direct use to
define a binary tree of peers, rather than of positions: If all peers
in the subtree of the address clockwise from $p_{i}$'s position have
positions which are in the subtree of $p_{cw}$, then we denote $p_{cw}$
the clockwise neighbor of $p_{i}$. Likewise, the counterclockwise
neighbor of $p_{i}$ is the peer $p_{ccw}$, whose position is the
fore-parent of all of the positions in the counterclockwise subtree
of $pos_{i}$ that are occupied by peers. 

It remains to define how messages can efficiently be routed from a
peer to its neighbors on the tree. The pseudocode of a protocol achieving
this is detailed in Alg. \ref{alg:CCWandCW}. To deliver messages
to the UP neighbor of a peer $p_{i}$, they are first addressed to
$UP\left[pos_{i}\right]$ and then continue being routed to the $UP\left[pos\right]$
of that address until they reach an address occupied by a peer. Clockwise
and counterclockwise messages are first routed to $CW\left[pos_{i}\right]$
and $CCW\left[pos_{i}\right]$. If they reach an address not occupied
by a peer, this is because the destination falls in the address space
segment of a peer $p_{j}$ occupying a different position. A new destination,
which is a step down the tree and away from that of $pos_{j}$, is
thus computed. If the destination address exhausts the address space,
the message is dropped.

Forwarding a message again and again until the destination is found
or the address space is exhausted is often wasteful and unnecessary.
Whenever the address of the destination position falls in the address
space of a peer who has a different position and is also a neighbor
of the sender, the message can be dropped. This is because the message
is doomed to be sent back and forth between the sender and the receiver
until eventually being dropped as there is no peer between them to
accept it. Fortunately, such communication patterns can easily be
avoided if the sender denotes as part of the message header the edge
of its address space in the direction in which the message is sent.
The recipient can then compare that to the edge of its own address
space segment. If the edges are the same, the message can be dropped.

Figure \ref{fig:Tree-on-DHT-1} illustrates this address scheme in
a DHT composed of just nine peers and an address space of eight bits.
For instance, peer number 5, whose address space segment is $\left(01110000,10011000\right]$,
takes the tree position 10000000, and so forth. Messages routed counterclockwise
from 9 first reach position 11010000, which is in the address space
segment of peer number 7. However, since the position of peer number
7 is 11000000, the message is then bounced clockwise to position 11011000,
which is occupied by peer number 8.

\begin{algorithm*}
\caption{\label{alg:CCWandCW}Local Binary Tree Routing}

\textbf{On downcall to SEND with message M and direction d:}

If d is upward then $dest\leftarrow UP\left[pos_{i}\right]$ and $edge\leftarrow null$

If d is counterclockwise then $dest\leftarrow CCW\left[pos_{i}\right]$
and $edge\leftarrow a_{i-1}$

If d is clockwise then $dest\leftarrow CW\left[pos_{i}\right]$ and
$edge\leftarrow a_{i}$

Make a downcall to SEND with the destination address $dest$ and the
message $\left\langle pos_{i},dest,edge,M\right\rangle $ using the
DHT

\textbf{On upcall DELIVER with the message $\left\langle origin,dest,edge,M\right\rangle $:}

If $dest=pos_{i}$ then call $ACCEPT$ with the message $M$ and finish.

If $dest$ is a fore parent of $origin$ then $newdest\leftarrow UP\left[dest\right]$
and $newedge\leftarrow null$

Else if $dest$ is in the clockwise subtree of $origin$ then

- If $edge=a_{i-1}$ then finish

- If $origin=pos_{i}$ then $newdest\leftarrow CW\left[dest\right]$
and $newedge\leftarrow a_{i}$

- Else $newdest\leftarrow CCW\left[dest\right]$ and $newedge\leftarrow a_{i-1}$

Else

- If $edge=a_{i}$ then finish

- - If $origin=pos_{i}$ then $newdest\leftarrow CCW\left[dest\right]$
and $newedge\leftarrow a_{i-1}$

- - Else $newdest\leftarrow CW\left[dest\right]$ and $newedge\leftarrow a_{i}$

- Make a downcall to SEND with the destination $newdest$ and the
message $\left\langle orig,newdest,newedge,M\right\rangle $
\end{algorithm*}

\subsection{Tree properties}

The properties of the tree which are the most important are its expected
maximal depth, the expected degree of internal nodes and the expected
stretch of a hop. The expected maximal depth is mostly important for
broadcast and convergecast applications, in which a message must traverse
the full depth of the tree. The expected degree determines the expansion
rate, which is central for repeated averaging algorithms such as gossip
algorithms. In any algorithm, actual performance is proportional to
the stretch -- the number of actual messages needed to deliver a message
from a tree node to its parent or its descendant.
\begin{lem}
The maximal depth of the tree is $O\left(\log N\right)$ where $N$
is the number of peers.\end{lem}
\begin{proof}
See \ref{sec:Proofs}.
\end{proof}
The stretch of the tree is defined in terms of the number of times
the tree routing protocol calculates a new destination for an UP message
and lets the DHT route the message. Each such message may require
up to $\log N$ IP messages. However, since the tree closely follows
the finger table logic, symmetric Chord peers will almost always have
a direct link to their CW, CCW and UP neighbors. Therefore, the number
of IP messages required for every DHT routing in symmetric Chord is
$O\left(1\right)$. Notice that messages in the CW and CCW direction
follow the same path on the tree as the UP message in the opposite
direction, and therefore have the same stretch. 
\begin{lem}
The expected stretch of the tree is a small constant.\end{lem}
\begin{proof}
See \ref{sec:Proofs}.
\end{proof}

\subsection{\label{sub:Detecting-neighbor-dynamics}Neighbor change notification}

The binary tree routing protocol in Alg. \ref{alg:CCWandCW} defines
neighbor relations logically and is therefore immune to peer dynamics.
Whenever peers join or leave the system the protocol simply reflects
the change by delivering messages according to the current tree structure.
However, some algorithms, including the one in the next section, still
require explicit notification when one of the tree neighbors changes.

The neighbor change notification protocol is based on the following
property of the binary tree routing protocol: Let $p_{i}$ be the
successor of $p_{i-1}$ and let their positions be $pos_{i}$ and
$pos_{i-1}$ respectively. If $p_{i-1}$ leaves the system then the
position of $p_{i}$ either remains $pos_{i}$ or changes to $pos_{i-1}$.
In the former case, the parent of $p_{i-1}$ becomes the parent of
its single direct descendant, if one exists. In the latter, $p_{i-1}$'s
former neighbors become the new neighbors of $p_{i}$ and the former
parent of $p_{i}$ becomes the parent of $p_{i}$'s former single
direct descendant. The same property can prove that it is sufficient
to alert those same five peers when $p_{i-1}$ joins the system.
\begin{lem}
The addition or removal of a peer $p_{i}$ can only affect the tree
connectivity of only five peers which are all tree neighbors of either
$p_{i}$ or its successor $p_{i+1}$.\end{lem}
\begin{proof}
See \ref{sec:Proofs}.
\end{proof}
This property can be used to provide alerts on any single local change
in topology. This is because when $p_{i-1}$ leaves or joins the system,
the DHT alerts its successor that its address space segment has changed.
Once $p_{i}$ is informed of the change in the address space segment
it is able to calculate the positions whose neighbors might have changed.
Hence, $p_{i}$ can route alert messages in all directions from those
positions.

\begin{algorithm*}
\caption{\label{alg:Change-Notification}Neighbor Change Notification}

\textbf{Definitions:} $Pos\left(a,b\right)$ is the position of a
peer whose address space segment is $\left(a,b\right]$

\textbf{On upcall NOTIFY that the predecessor address has changed
from $a_{i-2}$ to $a_{i-1}$ or vice-versa:}

Compute $pos_{fix}=Pos\left(a_{i-2},a_{i}\right)$ and $pos_{var}=\begin{cases}
Pos\left(a_{i-1},a_{i}\right) & Pos\left(a_{i-2},a_{i-1}\right)=pos_{fix}\\
Pos\left(a_{i-2},a_{i-1}\right) & Pos\left(a_{i-1},a_{i}\right)=pos_{fix}
\end{cases}$

Send the message $\left\langle ALERT,pos_{fix}\right\rangle $ in
direction UP, CW and CCW from $pos_{fix}$ using binary tree routing.

Send the message $\left\langle ALERT,pos_{vol}\right\rangle $ in
direction UP, CW and CCW from $pos_{vol}$ using binary tree routing.

\noindent \begin{raggedright}
\textbf{On upcall ACCEPT with the message $\left\langle ALERT,pos\right\rangle $:}
\par\end{raggedright}

If $pos$ is a fore-parent of $pos_{i}$ then $dir\leftarrow upward$

Else if $pos$ is in the clockwise subtree of $pos_{i}$ then $dir\leftarrow clockwise$

Else $dir\leftarrow counterclockwise$

Notify the application of a possible change of the neighbor in direction
$dir$
\end{algorithm*}

\section{\label{sec:Majority-voting}Majority Voting }

Given the infrastructure provided by the DHT overlay and by the binary
tree routing and change notification protocols, we next compare representative
local thresholding and gossip algorithms. We consider the simplest
computation task: a majority vote. However, the two algorithms we
choose are good representations of their respective families. We use
a variant of the local majority voting algorithm of Wolff and Schuster
\citep{MajorityRulej} and compare it to LiMoSense \citep{dynamicGossip},
which is a variant of the gossip averaging algorithm of Kempe et al.
\citep{KempeGossip}, suitable for dynamic data. Both algorithms were
slightly adapted, and are therefore described in the following subsections. 

The input for both algorithms is a single bit $x_{i}\in\left\{ 0,1\right\} $
at each peer $p_{i}$ and the computational task is to decide if on
average most bits are one or zero. We realistically assume that the
input of the peers can change at any moment, and thus that the algorithm
never terminates. The output of each peer is an ad-hoc assumption
on the majority.

\subsection{Local majority voting}

In local majority voting, every peer bases its output on the statistics
of votes it accepts from its tree neighbors -- namely: UP, CW, and
CCW. The peer stores for each of those neighbors two counter pairs:
$X_{UP,i}$ and $X_{i,UP}$ for the upward direction, $X_{CW,i}$
and $X_{i,CW}$ for the clockwise, and $X_{CCW,i}$ and $X_{i,CCW}$
for the counterclockwise direction. Each of those counter pairs counts
votes and the number of those votes which are of one. The counter
pair $X_{v,i}$ records the latest message received from direction
$v$, and $X_{i,v}$ the latest message sent to direction $v$. They
both are initially $\left(0,0\right)$. We conveniently denote $X_{\bot,i}=\left(x_{i},1\right)$
for the input of $p_{i}$. The knowledge of a peer is defined as the
sum of all its inputs $\mathcal{K}_{i}=\sum_{d\in\left\{ UP,CW,CCW,\bot\right\} }X_{d,i}$.
Whenever, according to its knowledge, the majority is of ones, $\left(1,-\frac{1}{2}\right)^{t}\mathcal{K}_{i}\geq0$,
the peer outputs one. Otherwise, it outputs zero.

To decide when and which messages it must send, the peer computes
for every direction $d\in\left\{ UP,CW,CCW\right\} $ the agreement
$\mathcal{A}_{i,d}=X_{d,i}+X_{i,d}$. A violation occurs when for
a direction $v\in\left\{ UP,CW,CCW\right\} $ the sign of the agreement
disagrees with the sign of the difference between the knowledge and
the agreement: $\left(1,-\frac{1}{2}\right)^{t}\mathcal{A}_{i,v}\geq0$
when $\left(1,-\frac{1}{2}\right)^{t}\left(\mathcal{K}_{i}-\mathcal{A}_{i,v}\right)<0$
or $\left(1,-\frac{1}{2}\right)^{t}\mathcal{A}_{i,v}<0$ when $\left(1,-\frac{1}{2}\right)^{t}\left(\mathcal{K}_{i}-\mathcal{A}_{i,v}\right)>0$.
Such violations can be triggered by initialization, by a change of
the peer's vote, or by an incoming message which changes one of the
$X_{d,i}$. 

To resolve a violation triggered by the agreement with a neighbor
in direction $v$, a peer can send a message containing information
on all of the votes received from neighbors in other directions. This
is done by computing $X_{i,v}\leftarrow\mathcal{K}_{i}-X_{v,i}$ and
sending $X_{i,v}$ to the neighbor in the direction $v$. Notice that
after this message is sent, $\mathcal{A}_{i,d}=\mathcal{K}_{i}$,
which resolves the violation.

When a neighbor in direction $v$ changes, the pairs $X_{i,v}$ and
$X_{v,i}$ no longer reflect messages sent to or received from the
current neighbor. Therefore, when a peer receives an alert of a change
in direction $v$, it sets $X_{v,i}$ to $\left(0,0\right)$ and sends
a message to that direction, which sets $\mathcal{A}_{i,v}$ once
more to $\mathcal{K}_{i}$. The change detection protocol alerts the
new neighbor as well. So the new neighbor will sending a message which
reflects its own knowledge. Once both peers send and accept those
messages, $\mathcal{A}_{i,v}$ is again equal to $\mathcal{A}_{v,i}$
and reflects an agreement between $p_{i}$ and its new neighbor.

Note that if $p_{i}$ does not have a neighbor in direction $v$,
then $X_{v,i}$ remains zero and does not affect $\mathcal{K}_{i}$
or the result. Messages sent by $p_{i}$ in direction $v$ would be
dropped by the binary tree routing protocol, but this would not be
indicated to $p_{i}$. We prefer wasting those messages to complicating
the protocol with NACK messages. Additionally, note that to support
the possibility of out of order message delivery, a sequential number
is attached to each outgoing message and a message is dropped when
it arrives after a message which was sent subsequently.

\begin{algorithm}
\caption{\label{alg:DHT-Majority-Voting}DHT Local Majority Voting}

\noindent \begin{raggedright}
\textbf{Input of peer $p_{i}$: }A vote $x_{i}\in\left\{ 0,1\right\} $
\par\end{raggedright}

\noindent \begin{raggedright}
\textbf{Data structure of $p_{i}$:}
\par\end{raggedright}

\noindent \begin{raggedright}
$X_{\bot,i}$ initializes to $\left(x_{i},1\right)$; $X_{UP,i}$,
$X_{i,UP}$,$X_{CW,i}$, $X_{i,CW}$,$X_{CCW,i}$, $X_{i,CCW}$, all
initialized to $\left(0,0\right)$, $seq$, $last_{UP}$, $last_{CW}$,
$last_{CCW}$ all initialized to 0.
\par\end{raggedright}

\noindent \begin{raggedright}
\textbf{Output of peer $p_{i}$:} One if $\left(1,-\frac{1}{2}\right)^{t}\mathcal{K}_{i}\geq0$
zero otherwise.
\par\end{raggedright}

\noindent \begin{raggedright}
\textbf{On change of $x_{i}$: }Set $X_{\bot,i}=\left(x_{i},1\right)$
and call test()
\par\end{raggedright}

\noindent \begin{raggedright}
\textbf{On an upcall to ACCEPT with a message $\left\langle X,seq\right\rangle $
from $Pos_{j}$:}
\par\end{raggedright}

\noindent \begin{raggedright}
Let $v\in\left\{ UP,CW,CCW\right\} $ be the direction of $Pos_{j}$
from $Pos_{i}$.
\par\end{raggedright}

\noindent \begin{raggedright}
If $seq>last_{v}$ then set $X_{v,i}\leftarrow X$, $last_{v}\leftarrow seq$,
and call test()
\par\end{raggedright}

\noindent \begin{raggedright}
\textbf{On an upcall to ALERT with direction $v$:} Set $X_{v,i}\leftarrow\left(0,0\right)$
and call Send$\left(v\right)$
\par\end{raggedright}

\noindent \begin{raggedright}
\textbf{Procedure test():}
\par\end{raggedright}

\noindent \begin{raggedright}
For $v\in\left\{ UP,CW,CCW\right\} $, if $\left(1,-\frac{1}{2}\right)^{t}\mathcal{A}_{i,v}\geq0$
and $\left(1,-\frac{1}{2}\right)^{t}\left(\mathcal{K}_{i}-\mathcal{A}_{i,v}\right)<0$
or $\left(1,-\frac{1}{2}\right)^{t}\mathcal{A}_{i,v}\geq0$ and $\left(1,-\frac{1}{2}\right)^{t}\left(\mathcal{K}_{i}-\mathcal{A}_{i,v}\right)<0$
then call Send$\left(v\right)$
\par\end{raggedright}

\noindent \begin{raggedright}
\textbf{Procedure Send$\left(v\right)$:}
\par\end{raggedright}

\noindent \begin{raggedright}
Let $X_{i,v}\leftarrow\mathcal{K}_{i}-X_{v,i}$, $seq\leftarrow seq+1$
\par\end{raggedright}

\noindent \raggedright{}Send a message $\left\langle X_{i,v},seq\right\rangle $
in direction $v$ using binary tree routing.
\end{algorithm}

\subsection{Gossip majority voting}

The gossip algorithm we use is a variant of LiMoSense \citep{dynamicGossip}.
To simplify the description and the experiments, we use the failure
free version, which does not handle joining and leaving of peers,
or unreliable messaging. We make one important adjustment to the algorithm:
instead of selecting the destination uniformly at random we select
uniformly from among the different destinations in the peer's finger
table. This is justified because in a DHT, following a random finger
$O\left(\log N\right)$ times will lead to a uniformly picked random
peer using just $O\left(\log N\right)$ messages. A second change,
which is semantic more than algorithmic, is that the output is quantized
to either zero or one, in line with the voting problem. A detailed
description of LiMoSense is not included here for lack of space.

\section{\label{sec:Experimentation}Experimental validation}

We conducted two sets of experiments to validate the usefulness of
our algorithms. The first experiment evaluates the performance of
the binary tree routing protocol in terms of the efficiency of the
tree it induces: the degrees of peers, their depth, and the stretch
-- the number of real messages required to send a message from a tree
node to its neighbor. The second compares the local majority voting
algorithm, which uses binary tree routing as the communication infrastructure,
to LiMoSense, which does not. The algorithms are compared in terms
of their scalability and response to stationary and non-stationary
changes in the data. 

We employed a standard peer-to-peer network simulator, peersim \citep{peersim}.
The simulator is efficient enough to simulate up to a million peers
in some experiments. We used reliable messaging and random network
delays of from one to ten simulation cycles. The objective of the
delay is not to approximate wall time but rather to decouple the peers
and avoid locked-step behavior. When using a Chord overlay, we use
an existing add-on to peersim. When using Symmetric Chord \citep{SChord},
we use our own variant, which initializes finger tables accordingly.
All measurements are averaged on ten random experiments, using different
random seeds.

\subsection{Tree Properties}

We investigated two key properties of the tree induced by the binary
tree routing protocol: The density, the depth, and the stretch.The
depth of the tree nodes is the distance from the root to each of them.
The depth is important mostly for applications which use global communication
such as broadcast and converge-cast because, for those applications,
the depth is proportional to the delay. As can be seen in Figure \ref{fig:Distribution-of-depth},
for a tree of $N$ peers, the first $\log N-2$ levels tend to be
completely full. The largest number of peers are at the $\log N$
level of the tree, and the reminder are at a small additional depth.
In none of the experiments we conducted, even with a million peers,
was a peer ever at a depth greater than $\log\left(N\right)+6$. We
conclude that the tree is extremely well balanced.

The stretch of a routing overlay is the number of actual messages
needed to deliver a message from a peer to its tree neighbor. This
metric assumes most of the cost of the protocol is associated with
application level routing decisions (i.e., finding the correct finger,
and so forth) and not with network delays. 

Figure \ref{fig:Number-of-hops} depicts the percentage of neighbors
at any given hop distance. It compares the results for a symmetric
Chord network of 10,000 and of 100,000 peers. The results are nearly
identical: 85\% percent of the peers are one or two hops away from
their tree neighbors. These results are then contrasted with a (non-symmetric)
Chord network of 10,000 peers. In that network the hop distance to
a neighbor is a combination of the hop distance to clockwise neighbors,
which is the same as that in symmetric Chord, and the hop distance
to a counterclockwise neighbor, which is the same as the distance
between any two random Chord peers. When using regular Chord overlay,
75 percent of the tree neighbors are within a hop distance of seven
or less. Although not as good, the average stretch is still well below
$\log N$. 

\begin{figure*}[p]
\caption{Tree depth and stretch}

\begin{minipage}[t]{0.5\textwidth}%
\subfloat[\label{fig:Distribution-of-depth}Distribution of peer depth]{

\includegraphics[width=1\textwidth]{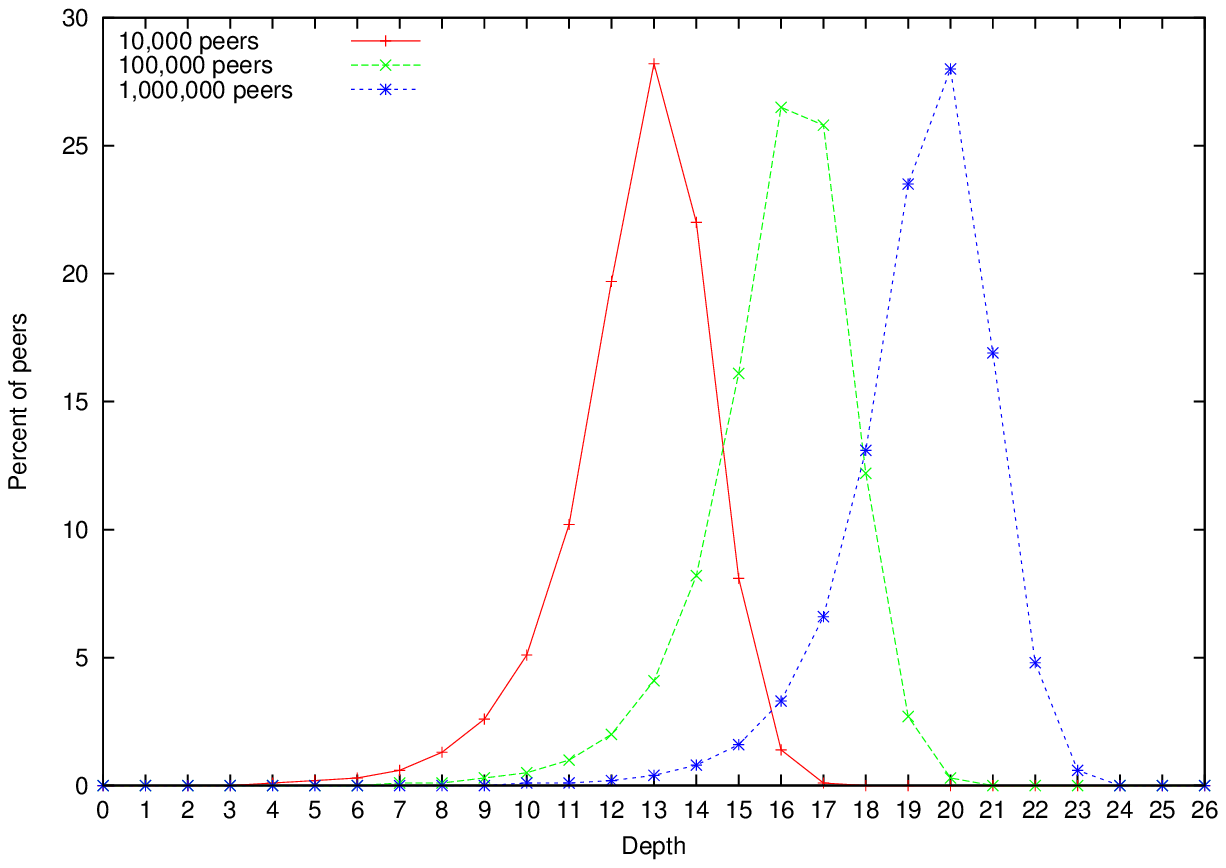}}%
\end{minipage}%
\begin{minipage}[t]{0.5\textwidth}%
\subfloat[\label{fig:Number-of-hops}Number of hops to tree neighbor]{

\includegraphics[width=1\textwidth]{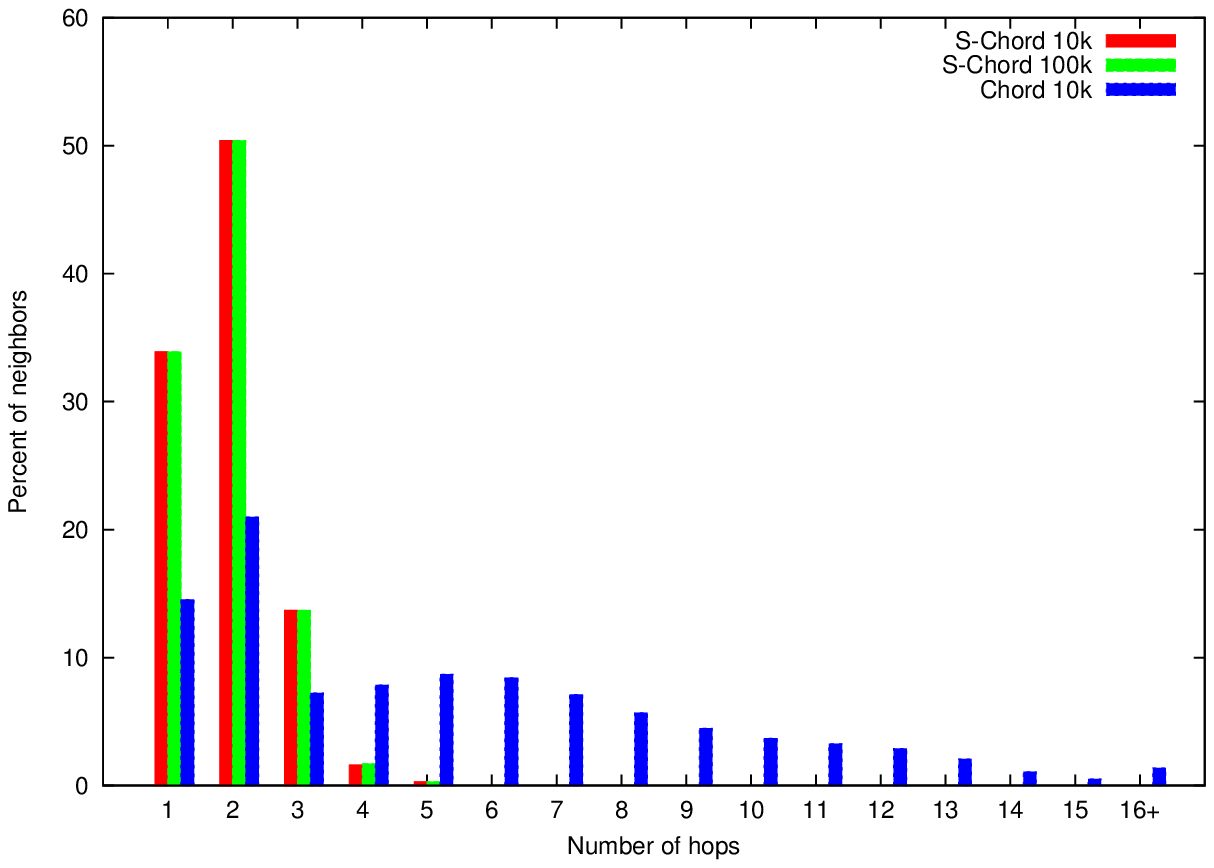}}%
\end{minipage}

\end{figure*}

\subsection{Majority Voting on DHT}

The second set of experiments compares local majority using the binary
tree routing protocol in the context of local majority voting with
majority based on gossip. We separate the experimented to between
those using static votes and those using stationary vote distributions.

\subsubsection{Static data}

An experiment with static data emulates a snapshot scenario of peer-to-peer
computation. In such scenarios, it is assumed that the input is a
distributed sample (i.e., snapshot) taken at very large intervals
-- large enough for the algorithm to stabilize between every two snapshots.
The goal of an algorithm in this state is to stabilize as quickly
and using as few messages as possible. We leave out experiments with
convergence time because of the difficulty of comparing the runtime
of a cycle-driven algorithm to an event driven one, because of space
considerations, and because the results of the two algorithms did
not differ notably in that respect.

The input of the peers is randomly set with an average $\mu_{pre}$.
Once all peers compute the same output of the majority function, the
input of some peers is randomly switched and the average is set to
$\mu_{post}$. At this point, the algorithm proceeds until all of
the peers once more compute the correct result. The number of messages
needed to reach this point is reported. Three very distinct cases
arise: $\mu_{pre}<\frac{1}{2}<\mu_{post}$, $\mu_{pre}<\mu_{post}<\frac{1}{2}$,
and $\mu_{post}<\mu_{pre}<\frac{1}{2}$. Other arrangements of $\mu_{pre}$,
of $\mu_{post}$, and of $\frac{1}{2}$ are symmetric because both
algorithms have no preference for a majority of ones or of zeros.
In the last of the three cases, convergence is instantaneous in both
algorithms because no peer ever outputs the wrong majority. We therefore
focus on the former two cases and experiment with two main arguments:
The scale -- number of peers, and the signal -- distances of $\mu_{pre}$
from $\mu_{post}$, and from $\frac{1}{2}$.

Figure \ref{fig:Convergence-with-static} depicts the number of messages
per peer required for each of the algorithms so that all peers compute
the correct majority on networks of 10,000 to 160,000 peers. The experiments
in Figure \ref{fig:Output-changes} depict the first case, with $\mu_{pre}$
and $\mu_{post}$ varied from 10\% vs. 90\% through to 40\% vs. 60\%.
The most evident outcome of these experiments is that local majority
is by far better than LiMoSense in this metric. 

The reason for the difference may be simple: in local majority, it
does not take long until only a few peers continue to exchange messages.
In LiMoSense, as well as in similar gossip algorithms, peers continue
to send messages periodically until the stopping criterion is reached.
In this experiment, the stopping criterion is that the last peer has
computed the correct result. However, gossip would remain inefficient
if other stopping criteria, such as a fixed number of cycles, or a
decrease of variance to some degree, are used. It is the data dependency
of the local thresholding algorithm which makes the difference.

\begin{figure*}[p]
\caption{\label{fig:Convergence-with-static}Messages until convergence with
static data}

\begin{minipage}[t]{0.5\textwidth}%
\subfloat[\label{fig:Output-changes}Output changes]{

\includegraphics[width=1\textwidth]{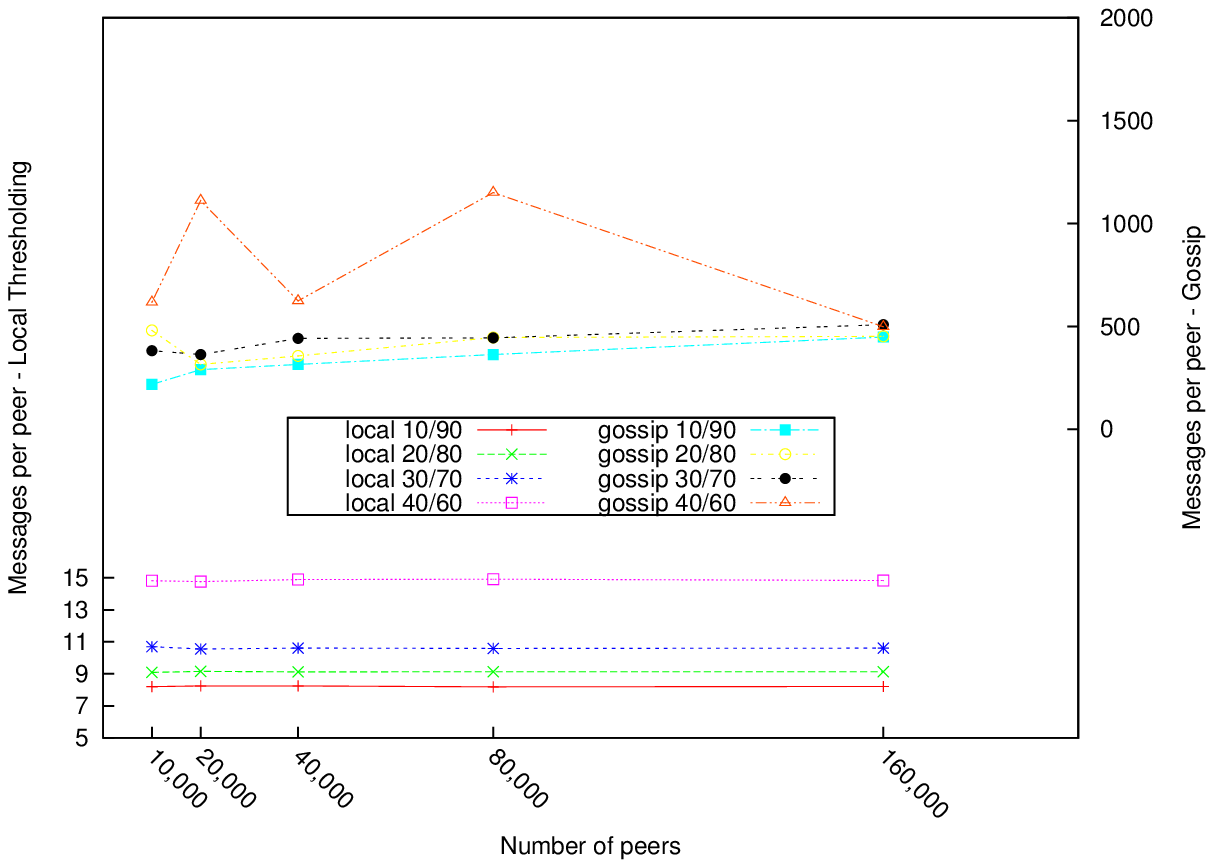}}%
\end{minipage}%
\begin{minipage}[t]{0.5\textwidth}%
\subfloat[\label{fig:Output-remains}Output remains the same]{

\includegraphics[width=1\textwidth]{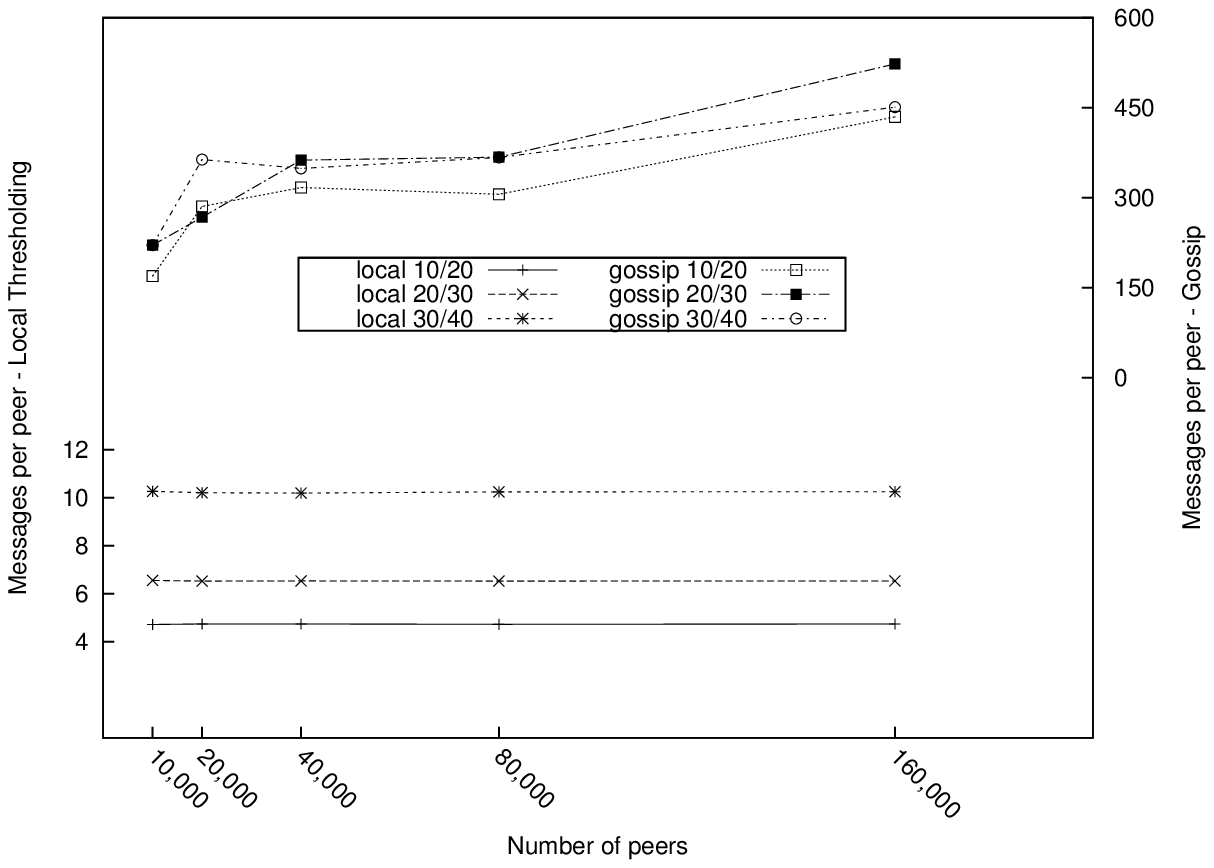}}%
\end{minipage}
\end{figure*}

\subsubsection{Stationary data}

In an experiment with stationary votes, a number of peers are randomly
picked at every given period and their vote is switched, keeping the
overall proportion of zero to one votes constant. We denote the fraction
of peers whose input changes at each average message delay (five simulation
cycles) the \emph{noise rate} and measure it in peers per million
per cycle (ppm/c). 

When inputs constantly change, convergence is impossible and convergence
cost becomes meaningless. Instead, it is the proportion of peers which
compute the correct outcome (i.e., the average accuracy) that matters,
and the ongoing communication costs required to preserve this level
of correctness. A second question is how well is the performance preserved
when the number of peer in the system grows.

Figures \ref{fig:Utility} and \ref{fig:Cost} depict the accuracy
and cost of local majority voting for networks of 10,000 to 160,000
peers and at various noise rates. As can be seen, regardless of the
noise rate, both average accuracy and average cost remain constant
when the system is scaled-up. Furthermore, even when more than one
peer in a thousand changes at every simulator cycle, the accuracy
remains above 90\%, and fewer than 2\% of the peers send a message
at every simulator cycle. 

\begin{figure}[p]
\caption{\label{fig:Scalability}Scalability of local majority on stationary
data}

\begin{minipage}[t]{0.5\textwidth}%
\subfloat[\label{fig:Utility}Local majority utility]{\includegraphics[width=1\textwidth]{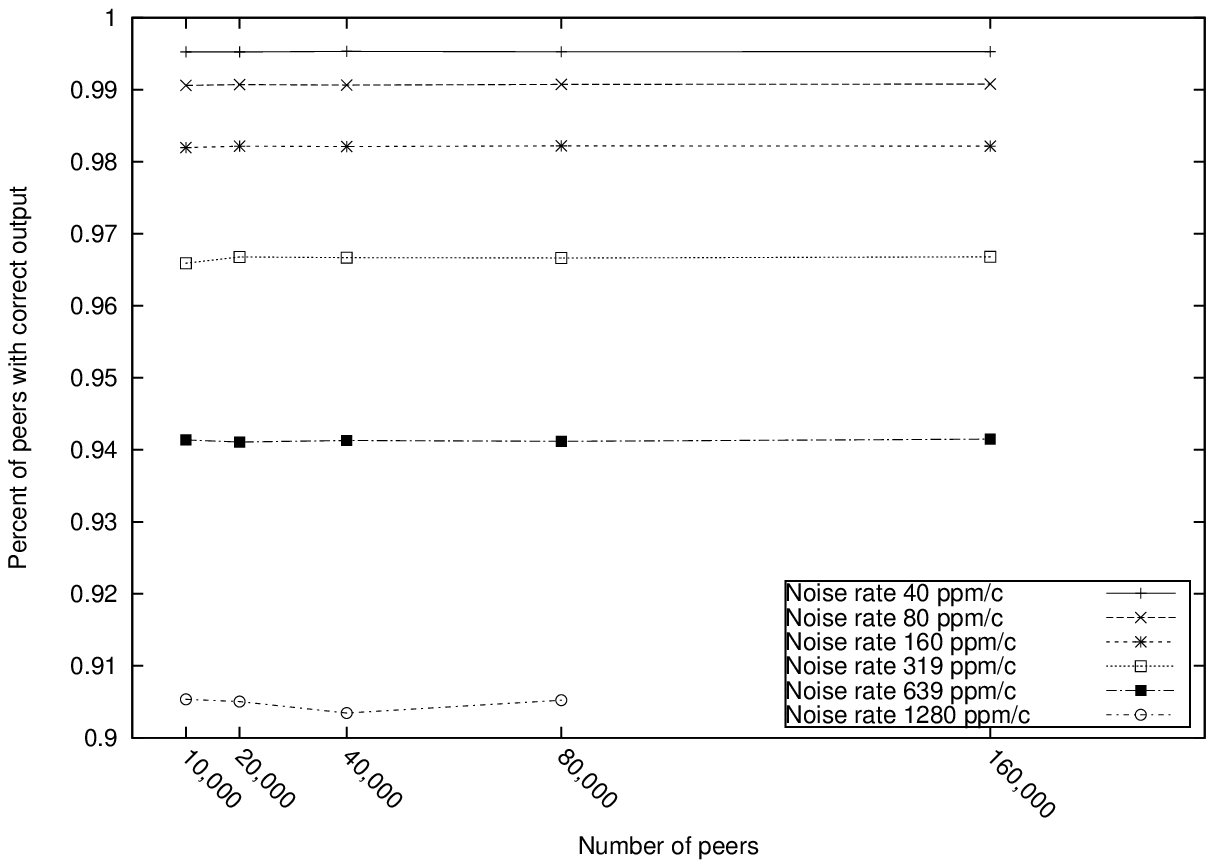}}%
\end{minipage}%
\begin{minipage}[t]{0.5\textwidth}%
\subfloat[\label{fig:Cost}Local majority cost]{\includegraphics[width=1\textwidth]{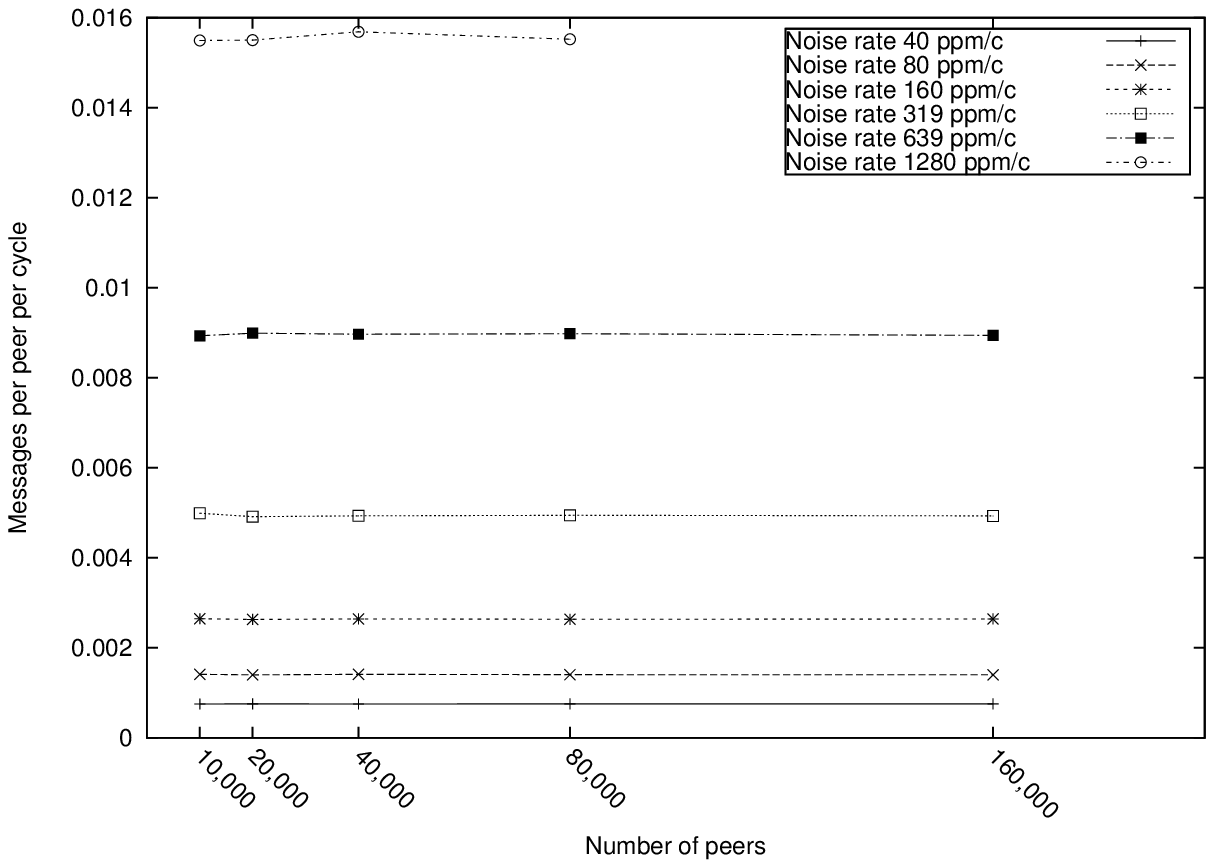}}%
\end{minipage}

\centering{}%
\begin{minipage}[t]{0.5\textwidth}%
\subfloat[\label{fig:Comparison-of-Local}Local vs. Gossip Majority Voting]{

\includegraphics[width=1\textwidth]{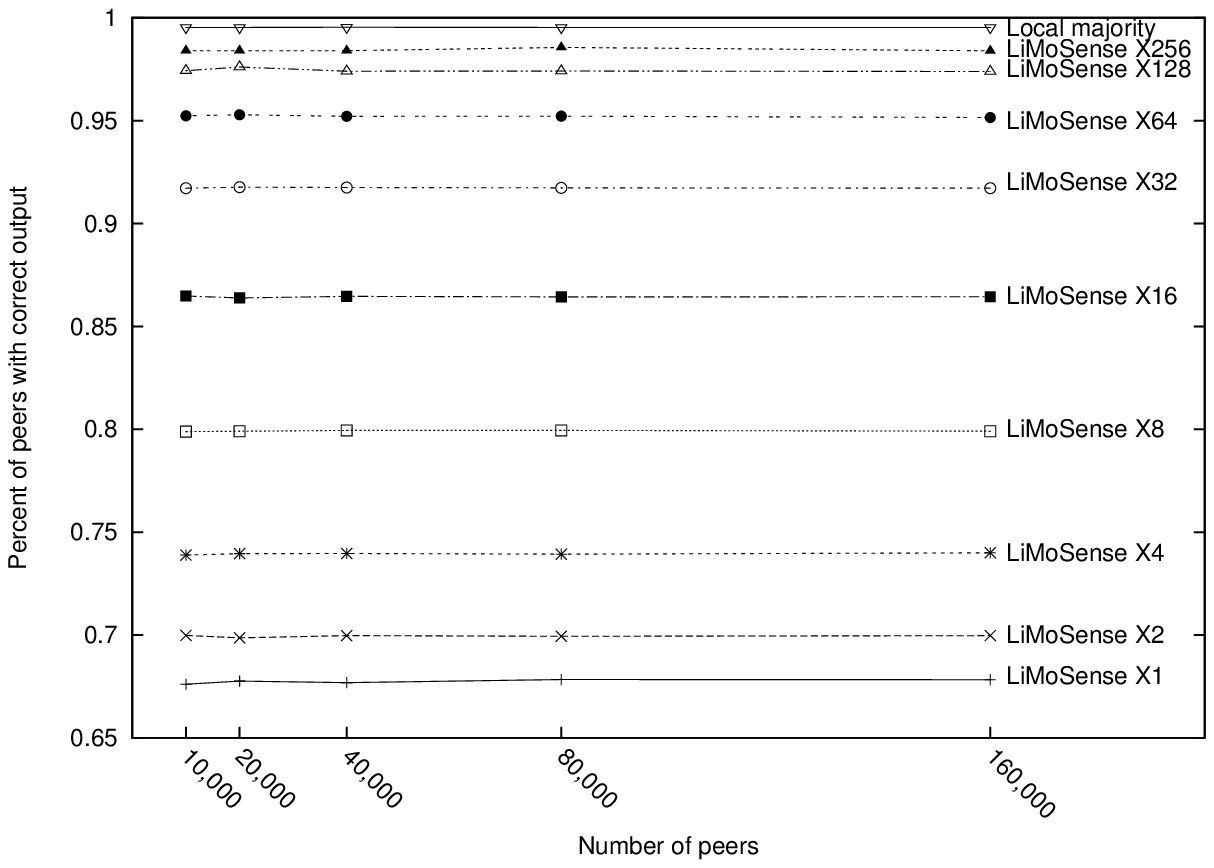}}%
\end{minipage}
\end{figure}

Finally, Figure \ref{fig:Comparison-of-Local} compares the performance
of local majority voting is compared to that of LiMoSense. To compare
the two algorithms on equal terms, the message overhead of LiMoSense
is set to exactly that of local majority voting. Then, LiMoSense is
allowed to send from twice that number to 256 times the number of
messages local majority voting sends. As can be seen, the utility
of LiMoSense does not degrade with scale. However, even when allowed
a number of messages which is eight orders of magnitude larger than
that of local majority voting, more than twice as many peers err,
on average, in LiMoSense. In terms of utility vs. cost, local majority
is overwhelmingly superior.

\section{\label{sec:Related-Work}Related Work}

The work described here relates to work in two areas: computation
and use of spanning trees in DHT overlays, and computation of majority
voting in those and other networks.

\subsection{Spanning trees in DHT overlays}

Bottom-up trees are discussed as part of the Scribe system \citep{Scribe},
in which peers are organized in groups and the peer whose address
is the closest to the groupId is the root. Reverse-path forwarding
allows broadcast in Scribe, but a single peer joining a group can
alter the parenthood relation. 

El-Ansary et al. \citep{DHTTree1} describe a partition based broadcast
tree that does not assure cycle freedom. Huang and Zhang \citep{dhttreeJ}
improve on that with a protocol that assures cycle freedom but in
which peer degrees vary from zero to $\log N$. Lately Huang and Zhang
\citep{DHTTree2} further improved their protocol with balanced DBT
in which the right and left descendants of a peer are, respectively,
the next peer and the peer responsible for the middle address of the
address space of the parent. Each peer then distributes the broadcast
to half the address space of the parent. Balanced DBT offers both
a bounded out-degree of two and a stretch that is typically one. The
binary tree routing protocol presented here further improves on balanced
DBT by removing the need for global partitioning of the address space.
Thus, it allows not only broadcast, but also convergecast, or multi-way
cycle free communication, which is the way local majority voting uses
it.

\subsection{Distributed majority voting}

The computation of the majority has long been a focal point of algorithms
intended for in-network computation. It was the subject of the first
local data mining algorithm \citep{MajorityRulej}, and is a straightforward
reduction of the push-sum gossip based protocol of Kempe et al. \citep{KempeGossip}.
Gossip based algorithms for majority voting were proposed \citep{GossipMajority,GossipVoting}.
However, they relate to the problem of limiting the space needed by
gossip and do not improve the messaging overhead beyond push-sup or
LiMoSense \citep{dynamicGossip}. 

Birk et al. \citep{ChargeFusion} suggested a local majority voting
algorithm for general networks. In that work, each ``1'' vote spans
a tree using the Bellman-Ford algorithm until is either nulled by
a ``0'' vote or it runs against the tree of another ``1'' vote.
The work has several limitations: Trees are data dependent, so they
have to be maintained when the data changes. Also, if multiple majority
votes are to be taken at once (as often happens in peer-to-peer data
mining), then different trees are computed for each vote, and expenses
accumulate. Bellman-Ford also incurs significant synchronization overhead
between the branches of the tree. In contrast, the local majority
voting algorithm described here relies on a binary tree protocol which
is data independent and only requires (local) maintenance on (local)
topology changes.

\section{\label{sec:Conclusions}Conclusions and future work}

For almost a decade since gossip based and local thresholding algorithms
were first described, the former remain the more practical and the
latter the more theoretically efficient. Our binary tree routing protocol
begins to bridge the gap. While interesting in itself, the protocol
is important because it permits seamless execution of any local thresholding
algorithm on a DHT overlay. The two kinds of algorithms can thus be
realistically compared. We believe the conclusion of such comparison
is beyond doubt: gossip based algorithms are by far inferior to local
thresholding algorithms for computation in DHT overlays.

The biggest challenge remaining is computation of local thresholding
algorithms on unstructured networks and on networks where communication
is noisy and asymmetric. Additionally, we see two interesting challenges
in implementing the binary tree protocol for other structured topologies,
and generalizing the protocol for trees of greater degree. Such generalization
may also serve as a means for controlling communication overhead which,
although low, is an artifact rather than an argument of current local
thresholding algorithms.

\bibliographystyle{elsarticle-num}
\bibliography{full}

\appendix

\section{\label{sec:Proofs}Proofs}
\begin{lem}
The address space segment associated with the peers whose positions
are the subtree below any peer $p_{i}$ is continuous.\end{lem}
\begin{proof}
Assume not then $p_{i}$ cannot be the root because its subtree is
all of the peers who, together, are associated with the entire address
space. Assume, without loosing generality that $p_{i}$ is in the
position $pos_{i}=p10^{k}$. For $p_{i}$ to have a discontinuous
address space there must be at least three peers $p_{cw}$, $p_{m}$,
and $p_{ccw}$ such that $p_{cw}$ is clockwise from $p_{m}$ which
is clockwise from $p_{ccw}$ and such that $p_{cw}$ and $p_{ccw}$
are in the subtree of $p_{i}$ but $p_{m}$ is not. Since $p_{cw}$
is in $p_{i}$'s subtree we know its position $pos_{cw}$ begins with
the prefix $p$ and the same goes for the position $pos_{ccw}$ of
$p_{ccw}$. Since the position correspond to an address in the peers
address space segment, we know all of the addresses in the address
space segment of $p_{m}$ begin with the prefix $p$. One of those
addresses correspond to $p_{m}$'s position $pos_{m}$ which must
therefore begin with the prefix $p$. Whatever that position is, by
applying the UP operator to it again and again we will reach $pos_{i}$.
Hence, $p_{m}$ is in $p_{i}$'s subtree, which contradicts the premise.\end{proof}
\begin{lem}
For any peer $p_{i}$ whose position is $pos_{i}$, one of the peers
which occupies positions in the subtree of position $CW\left[pos_{i}\right]$
(respectively, $CCW\left[pos_{i}\right]$) occupies a position which
is a fore-parent of the positions of all of the other peers in that
subtree.
\end{lem}
The lemma holds trivially if there are no peers or just one peer in
the subtree. Assuming there is more than one peer in the subtree,
let $pos_{p}$ be the lower common parent position of the positions
of all peers in the subtree. If $pos_{p}$ is occupied by one of those
peers, then the lemma is satisfied. Otherwise, $pos_{p}$ is not occupied
by a peer, possible only if it is in the address space segment of
a peer $p_{j}$ which occupies another position $pos_{j}$. Since
$pos_{p}$ is the lowest common parent, some of the other peers occupy
positions in $CW\left[pos_{p}\right]$ and some in $CCW\left[pos_{p}\right]$.
This means $p_{j}$ cannot be equal to $p_{i}$, since we know that
all the peers are in the subtree of $CW\left[pos_{i}\right]$ (respectively,
the subtree of $CCW\left[pos_{i}\right]$) . We are left with the
conclusion that $pos_{p}$ is in the address space segment of a peer
not in $p_{i}$'s subtree. However, this is in violation of Lemma
\ref{lem:continuous}.
\begin{lem}
The maximal depth of the tree is $O\left(\log N\right)$ where $N$
is the number of peers.\end{lem}
\begin{proof}
The binary tree is fully defined in terms of addresses regardless
of the positions actually occupied by peers. The clockwise and the
counter-clockwise subtrees of a peer at any address are span equal
address spaces. The peers, on the other hand, are randomly and uniformly
distributed in the address space. Hence, if the subtree of a peer
contains $k$ peers then the number of peers in every subtree is distributed
$Bin\left(\frac{1}{2},k\right)$. 

In a binary search tree built from random insertions, the distribution
of the number items in every subtree is uniform. It is known that
the maximal depth of a random binary search tree is roughly $4.3\log N$.
Since the probability that a subtree has more than $\frac{k}{2}+i$
nodes is higher in a random binary search tree than it is in the tree
induced by the protocol, the maximal depth of the tree which is induced
is expected to be smaller than $4.3\log N$.\end{proof}
\begin{lem}
The expected stretch of the tree is a small constant.\end{lem}
\begin{proof}
Call the destination address of the first hop the first address, and
that of the second hop the second address. Any address between that
of the initiator and the first and second addresses must be part of
the subtree of either the initiator or the first and the second peer,
respectively, because of Lemma \ref{lem:continuous}. If there are
more than two hops then the destination of the third hop must be in
the same address space segment as the first address. Or else, the
first address would be the highest in its address space segment, and
thus would be occupied by a peer which would become the parent of
the initiator. The same is true for the destination of the forth hop,
if there is one, and the second address. We conclude that if a message
in the UP direction makes more than two hops then those hops are between
two distinct peers -- the first and the second one -- and that eventually,
one of those peers must be the parent of the initiator.

The distance between the first and the second addresses must be larger
than the address space segment of the initiator. The distance between
the first and the third destinations is at least as large. If the
third destination is not the parent's address then the address space
segment which includes both the first and the third destinations must
be at least three time larger than the address space segment of the
initiator. Respectively, if the forth hop is not the last then the
address space segment of the second peer must be at least seven times
larger than the initiator's address space. In general, if the message
hops $k>2$ times then the address space segment of both the first
and the second peer must be at least $2^{k-2}-1$ larger than that
of the initiator.

It is known \citep{randomsegments} that the length of uniform random
segments is exponentially distributed. Given that the size of an address
space segment is $c$, the probability that the size of the consecutive
segments is $c\cdot2^{k}$ is the probability of sampling both values
from the exponential distribution, $Pr=\left(1-e^{-c\lambda}\right)e^{-c2^{k}\lambda}$.
For any constant $c$, this probability decreases double exponentially
in $k$. We conclude that the expected number of hops between a peer
and its parent is a constant not much greater than three.\end{proof}
\begin{lem}
The addition or removal of a peer $p_{i}$ can only affect the tree
connectivity of only five peers which are all tree neighbors of either
$p_{i}$ or its successor $p_{i+1}$.\end{lem}
\begin{proof}
When $p_{i}$ is added, the address space segment of its successor
$p_{i+1}$ is divided between $p_{i}$ and $p_{i+1}$. One of the
peers receives the address which previously corresponded with $pos_{i+1}$.
Clearly, if that peer is $p_{i}$ then the connectivity of the peers
which previously where the parent and direct descendants of $p_{i+1}$
changes, since $p_{i}$ now replaces $p_{i+1}$ as their neighbor.
The other peer, be it $p_{i}$ or $p_{i+1}$, receives a new position.
Call this peer $p_{new}$ and its position $pos_{new}$. Previous
to $p_{i}$ addition, $pos_{new}$ was not occupied by a peer because
the corresponding address was part of the address space of $p_{i+1}$
and that address space included a higher position -- $pos_{i+1}$.
The addresses between that corresponding with $pos_{new}$ and that
which corresponds to $pos_{i+1}$ are all in the address space of
either $p_{i}$ or $p_{i+1}$. Those addresses all correspond to positions
which are lower than the positions occupied by the two peers. Therefore,
$p_{new}$ can have at most one descendant. When a message was previously
routed up from that possible descendant, it was routed to $pos_{new}$
and then forwarded further up because $pos_{new}$ was part of an
address space belonging to a peer with a different position. Therefore,
whichever is the peer which now accept messages sent up from $p_{new}$,
that peer was the previous parent of $p_{new}$ sole possible descendant.
Because no other address space segment changes, no other peer changes
its position. Since we already enumerated the neighbors of $p_{i}$
and $p_{i+1}$, the connectivity of any peer other than those neighbors
does not change.\end{proof}

\end{document}